\documentclass[aps,prb,twocolumn,bibnotes]{revtex4-2} 

\usepackage{lipsum}

\usepackage{amsmath}
\allowdisplaybreaks
\usepackage{amssymb}
\usepackage{amsthm}
\usepackage{bbold,bm}
\usepackage[cal=boondox]{mathalfa}
\usepackage{graphicx}

\usepackage{hyperref}

\usepackage{scalerel}
\newlength\bshft
\def\mbold#1{\ThisStyle{\ooalign{$\SavedStyle#1$\cr%
  \kern-\bshft$\SavedStyle#1$\cr%
  \kern\bshft$\SavedStyle#1$}}}

\newcommand\vex[1]{\mathbf{#1}}

\def\bra#1{\mathinner{\langle{#1}|}}
\def\ket#1{\mathinner{|{#1}\rangle}}

\def\braket#1{\mathinner{\langle{#1}\rangle}}

\def\tr{\mathrm{tr}}

\def\id{\mathbb{1}} 

\def\Pexp{\mathrm{P}\hspace{-0.5mm}\exp}

\newtheorem{WNth}{Theorem}
\newtheorem{WNcr}{Corollary}[WNth]

\usepackage{xcolor}

\begin{document} 

\title{Wilson loop invariants and bulk-boundary correspondence in
higher-order topological insulators with anticommuting mirror symmetries} 

\author{Suman Aich}
\email{aichs@iu.edu}
\affiliation{Department of Physics, Indiana University, Bloomington, Indiana 47405, USA}
\affiliation{Quantum Science and Engineering Center, Indiana University, Bloomington, Indiana 47405, USA}

\author{Babak Seradjeh}
\email{babaks@iu.edu}
\affiliation{Department of Physics, Indiana University, Bloomington, Indiana 47405, USA}
\affiliation{Quantum Science and Engineering Center, Indiana University, Bloomington, Indiana 47405, USA}

\begin{abstract}
We investigate the higher-order bulk-boundary correspondence in a family of chiral-symmetric Bloch Hamiltonians with anticommuting mirror symmetries. These models generalize the $\pi$-flux square lattice, the prototypical topological quadrupole insulator, and include both separable and nonseparable models with extended and diagonal hopping. For separable systems, the product of subsystem chiral winding numbers correctly predicts the number of zero-energy corner states. However, this invariant fails in nonseparable models, motivating the development of new momentum-space diagnostics. We introduce gauge-independent mirror-filtered winding numbers for Wannier Hamiltonians, constructed by projecting mirror eigenstates onto the occupied subspace. Furthermore, by adapting periodicized Wilson lines from chiral Floquet theory to the case with momentum-dependent chiral operator, we define new invariants associated directly with Wannier gaps. These invariants provide a detailed characterization of Wannier band topology.  Our results clarify the interplay between chiral symmetry, mirror symmetries, and Wilson loops in higher-order topological phases and point to open challenges in formulating momentum-space invariants for general nonseparable models. 
\end{abstract}

{
\let\clearpage\relax
\maketitle
}

\section{Introduction}\label{sec:intro}
The study of higher-order topological phases has developed significantly since the formulation of quantized multipole insulators~\cite{Benalcazar_2017a,Benalcazar_2017b} followed by their experimental realizations in various platforms~\cite{Serra-Garcia_2018,Peterson_2018,Noh_2018,Imhof_2018,Fan_2019,Mittal_2019,Zhou_2020,Chen_2020,He_2020}. 
The prototypical model for a topological quadrupole insulator~\cite{Benalcazar_2017a,Benalcazar_2017b} is the same as the $\pi$-flux square model~\cite{Seradjeh_2008b}, which, along with other chiral-symmetric tight-binding models such as the Kekul\'e-patterned hexagonal model~\cite{Hou_2007}, have been the subject of earlier studies on the zero-energy states bound to vortex defects and the topological invariants characterizing them. The latter are the same invariants obtained from the index of the low-energy theory of Dirac fermions coupled to the vortex sector of a complex scalar field~\cite{Jackiw_1981a,Weinberg_1981,Cugliandolo_1989,Seradjeh_2008a,Seradjeh_2008d,Ryu_2009}. The study of defect bound states was later extended to all symmetry classes, both in static~\cite{Teo_2010} and periodically driven Floquet Hamiltonians~\cite{Yao_2017}.

Following the initial development of ideas on crystalline topological insulators~\cite{Fu_2011,Slager_2012,Neupert_2018,Kruthoff_2017,Khalaf_2018}, these and subsequent works~\cite{Schindler_2018,Langbehn_2017,Geier_2018,Trifunovic_2019,Khalaf_2021,Roy_2021} have established that crystalline symmetries can protect boundary states of co-dimension greater than one, giving rise to corner and hinge modes in the absence of edge or surface states corresponding to first-order topological invariants~\cite{Xie_2021,Yang_2024}. The original studies formulated the topological multipole insulators in terms of nested Wilson loops and the polarization of the Wannier bands~\cite{Fidkowski_2011,Alexandradinata_2014,Bradlyn_2019,Kooi_2019,Bouhon_2019,Hwang_2019,Henke_2021}. While nested Wilson loops capture multipole moments in models with appropriate crystalline symmetries, there is also a rich structure of higher-order topological phases and their invariants beyond multipole moments.

For example, in the $\pi$-flux square lattice, the prototypical topological quadrupole insulator, chiral symmetry protects a $\mathbb{Z}$-valued invariant that furnishes the bulk-boundary correspondence beyond the $\mathbb{Z}_2$-valued quadrupole moment of nested Wilson loops. The $\mathbb{Z}_2$ invariants are protected by mirror symmetries alone. Since this model has several symmetries, including discrete symmetries in the BDI class as well as mirror and rotational symmetries, it is natural to ask which symmetries protect and what invariants characterize the higher-order bulk-boundary correspondence.

In this work, we investigate a family of such models in two 
dimensions, constructed from tensor products of lower-dimensional chiral Hamiltonians and endowed with mirror and, in special cases, fourfold rotational symmetries. These models include both separable systems, whose topology is set by the independent winding of each subsystem, and nonseparable systems with mixed hopping processes. We examine their bulk and edge spectral structures, Wannier spectra of Wilson loops, and the invariants characterizing both.

For separable models, it has been shown~\cite{Hayashi_2018,Hayashi_2019,Okugawa_2019,Li_2018,Bomantara_2019,Yang2023} that the product of the one-dimensional chiral winding numbers of the lower-dimensional Hamiltonians of the tensor-product correctly predicts the number of zero-energy corner states, and remains robust so long as edges are gapped and the chiral symmetries are preserved. 

However, these invariants cannot be constructed in nonseparable models with diagonal hopping elements. Ref.~\cite{Benalcazar_2022} introduced real-space chiral winding numbers using a Bott index for the real-space representation of the quadrupole operator in a system with periodic boundary conditions, which yields a $\mathbb{Z}$ invariant matching the total number of corner-bound states, even in nonseparable models. However, it is not clear how to formulate this invariant in momentum space for the Bloch Hamiltonian, at least in part due to the fact that the real space quadrupole operator violates periodic boundary conditions, and so cannot be used with Bloch Hamiltonians. 

A central focus of this study is the role of chiral and mirror symmetries in defining appropriate bulk invariants. We construct gauge-independent winding numbers that characterize the Wannier Hamiltonians. We do so by introducing a gauge-independent ``mirror-filtered'' basis for the occupied bands, constructed by projecting mirror eigenstates onto the occupied subspace. Building on these constructions, we also define a complementary set of invariants directly associated with Wannier gaps at $\varepsilon = 0$ and $\pi$. These are formulated using periodicized Wilson lines, in analogy with invariants of Floquet systems. We extend the results from Floquet literature~\cite{Yao_2017} to incorporate the momentum-dependent chiral symmetry inherited from mirror projected onto the occupied subspace. 

The resulting invariants consistently relate to the mirror-filtered ones through the same linear combination found for chiral Floquet invariants~\cite{Asboth_2014}. Together, this set of gauge-independent topological invariants provide a detailed topological characterization of the Wannier Hamiltonians and their spectra. While these quantities furnish a bulk-boundary correspondence for Wilson loops, and in several cases also characterize the higher-order bulk-boundary correspondence, we have not found a simple mapping to do so in general. Thus, we shed light on the subtleties of relating nested Wilson loops and Wannier topology to higher-order bulk-boundary correspondence.

The paper is organized as follows. In Section~\ref{sec:models} we introduce the family of models and their symmetries. In Section~\ref{sec:Wgap}, we review the definitions of Wilson loops and present a comparison between the gap closings in their spectra and those in the bulk and edge spectra of the system. In Section~\ref{sec:symm}, we summarize the chiral invariants characterizing the Bloch Hamiltonian and construct the gauge-independent mirror-filtered invariants of the Wannier Hamiltonians and their spectra. 
We close in Section~\ref{sec:sum} with a summary and outlook. We also collect several relevant results on chiral winding numbers, including the consequence of various symmetries, in the Appendix.

\section{Model Hamiltonians}\label{sec:models}

We will first consider a family of two-dimensional (2D) Bloch Hamiltonians $H(\vex k)$ over the Brillouin zone $\text{BZ} = \{\vex k = (k_1,k_2) : -\pi < k_j \leq \pi \}$ ($j=1,2$), with 
\begin{equation}\label{eq:H0}
H(\vex k) = h_1(\vex k)\otimes \mathbb{1}  + c_1\otimes h_2(\vex k).
\end{equation}
Here, $h_j$ are chiral Hamiltonians with chiral operators $c_j$ (for which $c_j=c_j^\dagger=c_j^{-1}$ and $c_j h_j c_j = - h_j$) with additional mirror symmetries $m_{j l}$, reflecting $k_l \to - k_l$, that commute with each other $[m_{j1},m_{j2}]=0$.
We assume the commutation relations,
\begin{equation}\label{eq:CR12}
\{c_1, m_{11}\} = \{c_2, m_{22}\} = [c_1, m_{12}] = [c_2, m_{21}] =  0.
\end{equation}
Under this algebra, we can take $h_j = m_{jj} d_{je} + im_{jj}c_j d_{jo}$ where $d_{jo}$ is odd under $m_{jj}$ and both $d_{je}$ and $d_{jo}$ are otherwise even under reflections and commute with $c_j$ and $m_{jl}$. Typically, we simply have $m_{12}=m_{21}=\id$.  Table~\ref{table:hj_operators} summarizes the action of various $c_j$ and $m_{jl}$ operators on $h_j$ as well as their mutual commutation relations.

\begin{table}[t]
    \centering
    
    \bgroup
    \def\arraystretch{1.5}
    \begin{tabular}{c c c}
    \hline\hline
    {Operator} & $h_j(\vex k)$ & $m_{jr}$ \\
    \hline
    $c_j$ & $\{c_j, h_j(\vex k)\} = 0$ & $c_jm_{jr} = (-1)^{\delta_{jr}} m_{jr}c_j$ \\
    $m_{jl}$ & $m_{jl}h_j(\vex k) = h_j(\mathsf{m}_{l}\vex k) m_{jl}$ & $[m_{jl},m_{jr}] = 0$ \\
    
    \hline \hline
    \end{tabular}
    \egroup
    \caption{ Summary of the action of chiral and mirror operators on $h_j(\vex k)$ in Eq.~\eqref{eq:H0} and their commutation relations. The mirror-reflected momentum $(\mathsf{m}_{l}\vex k)_r = (-1)^{\delta_{lr}} k_r $}.
    \label{table:hj_operators}
\end{table}

The Hamiltonians~\eqref{eq:H0} are chiral under $C = C_1 C_2 = c_1\otimes c_2$, where $C_1=c_1\otimes \id$ and $C_2=\id\otimes c_2$ are the chiral operators of each of the terms in Eq.~\eqref{eq:H0} separately. Note that $[C,C_j] = 0$. Assuming $H$ is time-reversal symmetric under an antiunitary operator $\Theta$ with $\Theta^2=1$, such that $\Theta H(\vex k)\Theta = H(-\vex k)$, they would also preserve particle-hole symmetry $C\Theta$ and, thus, belong to the BDI class. In the following, we assume to work in the basis where $\Theta = K$, the complex conjugation.

The Hamiltonians~\eqref{eq:H0} are also mirror-symmetric, $M_j H(\vex k) M_j = H(\mathsf{M}_j \vex k)$, under $M_1 = m_{11}\otimes c_2 m_{21}$ and $M_2 = m_{12}\otimes m_{22}$, which all \emph{anticommute}:
\begin{align}
\{M_1,M_2\} = \{M_1,C\} = \{C,M_2\} = 0.
\end{align}
Here, $ \mathsf{M}_1\vex k = (-k_1,k_2)$ and $\mathsf{M}_2\vex k = (k_1,-k_2)$. Since $\vex k \to -\vex k$ under $\Theta$, we see that $m_{jl}$ and $c_j$, and therefore $M_j$, are real.

The inversion operator $I= iM_2M_1$ is also a symmetry $IH(\vex k)I = H(-\vex k)$. It commutes with the chiral operator $C$ and anticommutes with the mirror operators $M_j$. The factor of $i$ in this definition is needed to make sure $I$ is Hermitian, $I^\dagger= I = I^{-1}$. So, $\Sigma := I\Theta$ is an antiunitary symmetry, $\Sigma H(\vex k)\Sigma = H(\vex k)$, with $\Sigma^2 = -1$. By Kramers' theorem, the energy bands of $H$ are doubly degenerate at every $\vex k$.

For the special case $h_1(k_1,k_2) = h_2(k_2,k_1)\equiv h(\vex k)$, the models become $C_4$ symmetric with additional anticommuting diagonal mirror symmetries $M'_1 = e^{i\frac\pi2 IP_C} C_2$ and $M'_2 = e^{i\frac\pi2 I (1-P_C)} C_1$, where the projector $P_C =\frac12(1+C)$, such that $\mathsf{M}'_j(k_1,k_2)=(-1)^j(k_2,k_1)$ and
$
    M'_j H(\vex k) M'_j = H(\mathsf{M}'_j \vex k).
$
We note that, as can be directly inspected, $\{M'_1,M'_2\}=0$, however, interestingly, the diagonal mirror symmetries and the chiral operator \emph{commute}, $[C,M'_j] = 0$. Moreover, $I= i M'_2M'_1$, therefore, $M_j M'_l M_j = (-1)^{\bar j} M'_{\bar l}$ and $M'_j M_l M'_j =(-1)^j M_{\bar l}$ with $\bar j \neq j$. In this case the symmetry operator $R = e^{i\pi/4} M_1 M'_1 = e^{i\pi/4} M_2 M'_2$ is the $\pi/2$ rotation such that $R^4 = 1$ and 
$R^\dagger M_j R = (-1)^{j} M_{\bar j}$, and similarly $R^\dagger M'_j R = (-1)^{j} M'_{\bar j}$.  Table~\ref{table:H_operators} summarizes the action of the $C$, $M_j$ and $M'_{j}$ operators on $H$ as well as their mutual commutation relations.

\begin{table*}[t]
    \centering
    \bgroup
    \def\arraystretch{1.5}
    \begin{tabular}{c c c c}
    \hline\hline
    {Operator} & $H(\vex k)$ & $M_l$ & $M_l'$ \\
    \hline
    $C$ & $\{C, H(\vex k)\} = 0$ & $\{C, M_l\} = 0$ & $[C, M_l'] = 0$ \\
    $M_j$ & $M_j H(\vex k) = H(\mathsf{M}_j\vex k) M_j$ & $\{M_j, M_l\} = 2\delta_{jl}$ & $M_j M_l' = (-1)^{\bar j}M_{\bar l}'M_j$ \\
    $M_j'$ & $M'_j H(\vex k) = H(\mathsf{M}'_j\vex k) M'_j$ & $M'_j M_l = (-1)^j M_{\bar l} M_j'$ & $\{M_j',M_l'\} = 2\delta_{jl}$ \\
    \hline \hline
    \end{tabular}
    \egroup
    \caption{ Summary of the action of chiral and mirror operators on $H(\vex k)$ and their commutation relations. The mirror-reflected momenta $(\mathsf{M}_j\vex k)_l = (-1)^{\delta_{jl}} k_l$ and $(\mathsf{M}'_j \vex k)_l = (-1)^j k_{\bar l}$. The relation for $M'_j$ and $H(\vex k)$ holds in a $C_4$ symmetric model.}
    \label{table:H_operators}
\end{table*}

As a concrete example, we choose $h_j = d_{je}\sigma_x + d_{jo}\sigma_y$ and denote 
$d_{je}+id_{jo} \equiv d_j$ to be an even function of $k_{\bar j}$. Then, the symmetry operators are  $c_1 = c_2 = \sigma_z$, $m_{j \bar j} = \mathbb{1}$, $m_{j j} = \sigma_x$, so $C = \sigma_z\otimes\sigma_z$, and
\begin{align}
M_1 &= \sigma_x \otimes \sigma_z, \quad
M_2 = \mathbb{1}\otimes\sigma_x, \quad
I = \sigma_x\otimes \sigma_y, \\ 
M'_1 &= 
\frac12(\id\otimes \sigma_z - \sigma_z \otimes \id - \sigma_x \otimes \sigma_x + \sigma_y \otimes \sigma_y ), \\
M'_2 &= 
\frac12(\id \otimes \sigma_z + \sigma_z \otimes \id + \sigma_x \otimes \sigma_x + \sigma_y \otimes \sigma_y).
\end{align}
In particular, we study models of the form,
\begin{equation}\label{eq:ex}
    d_j = 1+f_j(n_{\bar j}k_{\bar j}) + [1-f_j(n_{\bar j}k_{\bar j})] e^{i n_j k_j},
\end{equation}
with $n_j\in\mathbb{Z}$ and $f_j$ real-valued even functions of their argument. For $n_1=n_2 \equiv n$ and $f_1(k)=f_2(k) \equiv f(k)$, we have $d_1 = d_2 \equiv d$ and the model is $C_4$ symmetric. Specifically, we choose
\begin{equation}\label{eq:ff}
f_j(k) = \bar f_j + \delta f_j \cos k.
\end{equation}
In the following, we will assume $|f_j(k)|<1$ for simplicity. With this choice $|d_j(k_j=\pi/n_j)| < |d_j(k_j=0)|$.

The case with $\delta f_j=0$ is an example of a ``separable model" when $h_j$ depends only on the component of momentum $k_j$. This entails hopping only to the $n_j$ neighbor in each direction $j$. For $n_1=n_2=1$ in 2d, we find the $\pi$-flux square lattice model (a 2d extension of the Su-Schrieffer-Heeger model)~\cite{Seradjeh_2008b} and a topological quadrupole insulator~\cite{Benalcazar_2017a}, the first example of a higher-order topological insulator. 

The general case of ``nonseparable models,'' where $f_j$ is not constant, includes mixed (diagonal) hopping terms. Such nonseparable models with longer-ranged hopping can arise, for example, as the Floquet Hamiltonians of periodically driven systems~\cite{Rodriguez-Vega_2019,Zhou_2025}.  We show the schematics of the tight-binding model in Eq.~\eqref{eq:ff} with $\delta f_j=0$ (separable) and $\delta f_j\neq0$ (nonseparable) in Fig.~\ref{fig:unit cell}(a) and~\ref{fig:unit cell}(b), respectively. We also compare our results to a nonseparable model in the literature~\cite{Benalcazar_2022}, with
\begin{equation}\label{eq:bc}
d_j = 1+ \bar f_{j} + [1-\bar f_{j}-\delta f_j \cos(k_{\bar j})] e^{ik_j} + (1-\bar g_{j}) e^{2ik_j},
\end{equation}
where nearest ($1\pm \bar f_{j}$), next-nearest ($\pm \delta f_j$), and next-next-nearest ($1-\bar g_j$) hopping elements are combined. 

\begin{figure}[b]
    \centering
    \includegraphics[width=\linewidth]{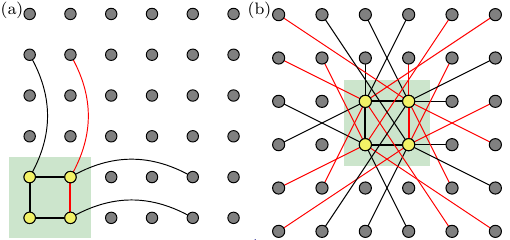}
    \caption{ Schematics of the tight-binding (a) separable model in Eq.~\eqref{eq:ff} ($\delta f_1 = \delta f_2=0$) with $n_1=n_2=2$ and (b) non-separable model in Eq.~\eqref{eq:ff} ($\delta f_1 = \delta f_2\neq0$) with $n_1=n_2=1$. The unit cell is the green square. Positive and negative hopping amplitudes are shown in black and red, respectively.
    }
    \label{fig:unit cell}
\end{figure}

\section{Wilson loops and edge spectra}\label{sec:Wgap}
Denoting the Bloch states in band $n$ as $\ket{u_n(\vex k)}$, the Wilson line operator in the subspace $S$ spanned by a subset of bands over a line $\ell \subset \text{BZ}$ starting at $\vex k_0$ and ending in $\vex k_1$ is the unitary
\begin{align}
 \overline{W}_{S\ell}(\vex k_0,\vex k_1) 
    = \prod^{\xrightarrow{~~~}}_{\vex q \in \ell} P_S(\vex q) 
    = \Pexp \left[i \int_\ell \vex{A}(\vex q)\cdot d\vex q\right],
\end{align}
where $d\vex q$ is the tangent to the loop, the projection $P_S(\vex k) = \sum_{n} \ket{u_n(\vex k)}\bra{u_n(\vex k)}$, $\Pexp$ is the path-ordered exponential, and $\vex{A}(\vex q) = \sum_{mn}\ket{u_m(\vex q)} [\mathbb{A}(\vex q)]^{mn} \bra{ u_n(\vex q)}$ is the non-Abelian Berry connection with projected matrix elements $[\mathbb{A}(\vex q)]^{mn} = -i\braket{ u_m(\vex q)|\partial_{\vex q} u_n(\vex q)}$. Note that we take the order of the operator product to be from left to right. For a loop $\ell$ starting at $\vex k_0$ and ending at $\vex k_1 = \vex k_0 + \vex G \equiv \vex k_0$, where $\vex G$ is a reciprocal lattice vector, we denote the Wilson loop operator $W_{S\ell}(\vex k_0) = \overline{W}_{S\ell}(\vex k_0,\vex k_0+\vex G)$.

We may express the Wilson line operator as $\overline{W}_{S\ell} = \sum_{mn}\ket{u_m(\vex k_0)} [\overline{\mathbb{W}}(\vex k_0, \vex k_1)]^{mn} \bra{ u_n(\vex k_1)}$ with the projected Wilson line \emph{matrix}  $\overline{\mathbb{W}}= \Pexp \left[i \int_\ell \mathbb{A}\right]$. Under a unitary gauge transformation of the Bloch-band bundle in $S$, $\ket{u'_n(\vex k)} = \sum_m [\mathbb{g(\vex k)}]^{mn} \ket{u_m(\vex k)}$, the Berry connection and Wilson line matrices transform as
\begin{align}
\mathbb{A}'(\vex k) &= \mathbb{g}^\dagger(\vex k) \mathbb{A}(\vex k) \mathbb{g}(\vex k)  - i \mathbb{g}^\dagger(\vex k) \partial_{\vex k} \mathbb{g}(\vex k) \label{eq:A'A} \\
\overline{\mathbb{W}}'(\vex k_0,\vex k_1) &= \mathbb{g}^\dagger(\vex k_0) \overline{\mathbb{W}}(\vex k_0,\vex k_1) \mathbb{g}(\vex k_1).
\end{align}
Thus, while the Wilson line matrix is transformed by the gauge, the  Wilson line operator itself is independent of the gauge. Note that the Wilson loop matrix is mapped unitarily under the gauge transformation, $\mathbb{W}'(\vex k) = \mathbb{g}^\dagger(\vex k) \mathbb{W}(\vex k) \mathbb{g}(\vex k)$. 

Taking the occupied bands, $S \equiv \text{image}(P)$, where $P$ is the projector to the occupied bands (i.e. with negative energies), and the loops $\ell_j=\{\vex k : -\pi < k_j \leq \pi\}$, we denote $W_{S\ell_j}\equiv W_{j}$. Similarly, since the projector on the unoccupied bands $1-P = CPC$, we find the Wilson loop operator of the unoccupied bands to be $CW_{j}C$.

We would also like to define the Wannier Hamiltonians $A^W_{\ell}(\vex k) \stackrel{?}{=} -i\ln W_{\ell}(\vex k)$. We must note, however, that since $W_j = P W_j P$ is a projected operator, it is not unitary and its logarithm cannot be defined in the full Hilbert space. Instead, since the Wilson loop matrix is unitary, we may write $A_\ell^W(\vex k) = \sum_{mn} \ket{u_n(\vex k)} [\mathbb{A}_\ell^W(\vex k)]^{nm}  \bra{u_m(\vex k)}$ and define the Wannier Hamiltonian matrix
\begin{equation}
\mathbb{A}^W_{\ell}(\vex k) = -i\ln \mathbb{W}_\ell(\vex k),    
\end{equation}
directly in the projected space.

Wannier Hamiltonians have been argued to have the same topology as the original Hamiltonian in cylindrical geometry with an edge along $k_{\bar j}$~\cite{Fidkowski_2011}. 
To investigate this equivalence, we first  compare the spectrum of the original and the Wannier Hamiltonians. Since the Wannier Hamiltonian is a modular operator, it has two relevant gaps, $\varrho_{j\varepsilon}$, at Wannier eigenvalues $\varepsilon = 0$ and $\pi$, which mark the Wannier centers $\varepsilon/2\pi$ in the unit cell. In Figs.~\ref{fig:BEWsep}-\ref{fig:BEW-BC}~\cite{Aich_2025c}, we show the bulk and edge spectral gaps of the models in Eqs.~\eqref{eq:ff} and~\eqref{eq:bc}, calculated from the Bloch Hamiltonians (periodic boundary conditions) and their unidirectional discretization (cylindrical boundary conditions), respectively, along with the Wannier gaps of the corresponding Wilson loops. 

\begin{figure}
    \centering
    \includegraphics[width=\linewidth]{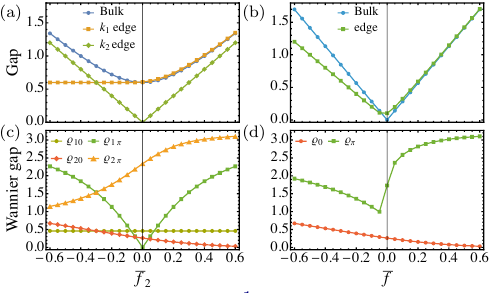}
    \caption{(a),(b) Bulk and edge gaps and(c),(d) Wannier gaps for the separable model in Eq.~\eqref{eq:ff} ($\delta f_1 = \delta f_2=0$) with $\bar f_1=-0.3$ and variable $\bar f_2$ (left panels), and $C_4$ symmetric model with variable $\bar f$ (right panels). In both cases $n_1=n_2=2$.  The bulk and Wannier gaps are calculated with periodic boundaries, and the edge gaps are calculated with cylindrical boundaries with 30 unit cells in the open direction.
 }
    \label{fig:BEWsep}
\end{figure}

\begin{figure}
    \centering
    \includegraphics[width=\linewidth]{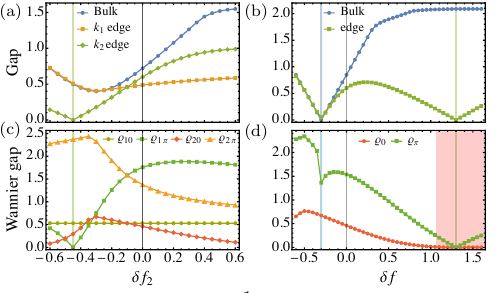}
    \caption{(a),(b) Bulk and edge gaps and (c),(d) Wannier gaps for the nonseparable model in Eq.~\eqref{eq:ff} with $\delta f_1 = -0.1$ and variable $\delta f_2$ (left panels), and for the $C_4$ symmetric model with variable $\delta f$ (right panels). In both cases $n_1=n_2=1$, and $\bar f_1=\bar f_2=-0.3$. The vertical blue and green lines show bulk and edge gap closings, respectively. The shaded region in (d) marks the vanishingly small value of the Wannier gap.  The bulk and Wannier gaps are calculated with periodic boundaries, and the edge gaps are calculated with cylindrical boundaries with 30 unit cells in the open direction.
    }
    \label{fig:BEWnonsep}
\end{figure}

Indeed, the spectral edge gap and the Wannier gap close together. Interestingly, in all the cases we have studied the Wannier gap closing that accompanies the spectral edge gap closing along $k_{\bar j}$ is $\varrho_{j\pi}$~\cite{Khalaf_2021}.
However, the relationship between the bulk and edge, and Wannier spectra is more complex. For example, while bulk gap closings are not accompanied by Wannier gap closings, they coincide with kinks in the Wannier spectrum. Also, the Wannier gaps can become anomalously small (shaded areas in Figs.~\ref{fig:BEWnonsep} and~\ref{fig:BEW-BC}) or show kinks and even gap closings unrelated to the spectral gaps [dashed lines in Fig.~\ref{fig:BEW-BC}(d)].

Now, we turn to characterizing the topology of the Hamiltonain and its relation to Wilson loops.

\section{Topological invariants and bulk-boundary correspondence}\label{sec:symm}

\subsection{Chiral winding numbers}

A chiral Hamiltonian $h(\vex k)$ with chiral operator $c$ is characterized by a topological winding number over a loop $\ell$ in the Brillouin zone, given by
\begin{equation}
w_{\ell}[h,c] = \frac1{2\pi i} \oint_\ell \partial_k \ln\det[d(k)] dk,
\end{equation}
where $d(k)$ is the off-diagonal block of $h(k)$ in the chiral basis, where $c = \sigma_z\otimes \id$. This winding number can also be expressed as
\begin{equation}
w_{\ell}[h,c] = \frac1{2\pi i} \oint_\ell \tr[p_c \, \underline h\,\partial_k \underline h] dk,
\end{equation}
where $\underline h = 1 - 2p$ is the flattened Hamiltonian with $p$ the projector to the occupied band with negative energies, and $p_c = \frac12(1+c)$ is the projector to the positive chiral eigenspace. This is proved in the Appendix, where we collect several relevant results on winding numbers.

\begin{figure}
    \centering
    \includegraphics[width=\linewidth]{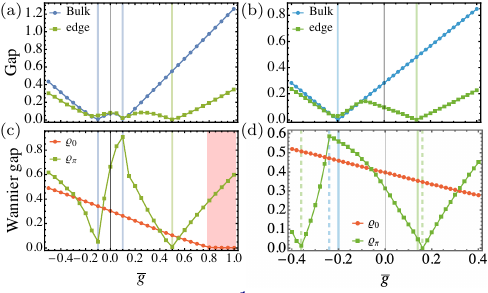}
    \caption{(a),(b) Bulk and edge gaps and (c),(d) Wannier gaps for the $C_4$ symmetric model~\eqref{eq:bc} with variable $\bar g$ and $\bar f =0.11$, $\delta f=1.11$ (left panels), and $\bar f =-0.3$, $\delta f=0.6$ (right panels). The vertical blue and green lines show bulk and edge gap closings, respectively. The vertical dashed lines in (d) show the kinks observed in the Wannier gap. The shaded region in (c) marks the vanishingly small value of the Wannier gap.  The bulk and Wannier gaps are calculated with periodic boundaries, and the edge gaps are calculated with cylindrical boundaries with 30 unit cells in the open direction.
    }
    \label{fig:BEW-BC}
\end{figure}

The commutation relations~\eqref{eq:CR12} ensure that $h_1$ and $h_2$ can have nontrivial winding numbers $\nu_j(k_{\bar j}) = w_{\ell_j}[h_j,c_j]$. 
In our models, $\nu_j = \frac1{2\pi} \oint (\partial_j \arg d_j) dk_j$, where we have denoted $\partial_j \equiv \partial/\partial k_j$. For separable models, $\nu_1$ and $\nu_2$ are constant, and the topological invariant $\nu := \nu_1\nu_2$ furnishes a higher-order bulk-boundary correspondence with $4|\nu|$ zero-energy corner-bound states ($|\nu|$ per corner)~\cite{Okugawa_2019, Hayashi_2018, Hayashi_2019}. 
The invariant $\nu=\nu_1\nu_2$ is protected as long as the chiral symmetry is preserved and the energy gap of the edge states with open boundary conditions is not closed.

In $C_4$ symmetric models the Hamiltonian and the diagonal mirror symmetries commute at every point along $\ell'_{j}$, the invariant line of the  diagonal mirror operators, $M'_{j} H(\vex k_{*j}) M'_{j} = H(\vex k_{*j})$ where $\vex k'_{*j} \in \ell'_j = \{(q,(-1)^jq): -\pi < q \leq \pi\}$. Therefore, although $w_{\ell'_j}[H,C]=0$ since $M'_{\bar j}$ is a symmetry that commutes with $C$ and inverts $\ell'_j$ (see Corollary~\ref{cr:MC} in the Appendix), we may define a ``mirror-graded'' invariant $\nu'_{j} = \frac12 w_{\ell'_j}[H, M'_j C]$. Note that $w_{\ell'_j}[H, M'_j C]$ is guaranteed to be even by the vanishing of $w_{\ell'_j}[H,C]$, since the latter can be thought of the sum of the winding numbers of $H$ with respect to $C$ in the $\pm1$ eigenspaces of $M'_j$ and the former as their difference. The bulk invariant $\nu'_j$ corresponds to the number of counter-propagating pairs of edge modes with opposite eigenvalues of $M'_j$ in cylindrical boundary conditions with an open edge perpendicular to $\ell'_j$. In an open geometry, it can be shown that $\nu'_j$ corresponds to the difference between the number of corner bound states with opposite eigenvalues of $M'_j$~\cite{Rodriguez-Vega_2019,Aich_2025a}.

Which invariant, $\nu$ or $\nu'$, furnishes the correct bulk-boundary correspondence in a $C_4$ symmetric model? We showed in Ref.~\cite{Aich_2025a}, that the correct bulk-boundary correspondence is furnished by $\nu$ and not $\nu'$: the number of corner-bound states is not captured by the diagonal mirror eigenvalues. Only $|\nu'|$ out of $|\nu|$ corner states have the same diagonal mirror eigenvalue, and each of the remaining $\frac12|\nu'|(|\nu'|-1)$ \emph{pairs} of bound states have opposite diagonal mirror eigenvalues. However, all $|\nu| = |\nu'|^2$ bound states have the same chiral eigenvalue.

The general case of non-separable models, with or without $C_4$ symmetry, presents a further challenge for establishing the bulk-boundary correspondence since $\nu_1$ and $\nu_2$ may no longer be constants. Each $\nu_j$ vanishes at a Dirac point of $h_j$, and the nonzero values of $\nu_j(k_{\bar j})$ correspond to flat, zero-energy edge bands of $h_j$ in open boundary conditions.

\subsection{Mirror-filtered invariants of the Wannier Hamiltonian}

Nested Wilson loops, that is, the Wilson loop $W_{\bar j j}$ of the Wannier spectrum of $W_j$, are shown to characterize the multipole moments of the topological quadrupole insulator. However, this invariant only admits a $\mathbb{Z}_2$ value and is not sufficient to characterize the $\mathbb{Z}$-valued higher-order topology of our models. It is natural to ask whether, instead of nested Wilson loops, the topology of $W_j$ can be characterized and be related to the higher-order bulk-boundary correspondence.

Under mirror symmetry $M$, such that $MH(\vex k) = H(\mathsf{M}\vex k)M$, the Wilson line operators satisfy the relation
\begin{align}\label{eq:MWlineM}
    M \overline{W}_{\ell}(\vex k_0, \vex k_1)M = \overline{W}_{\mathsf{M}\ell}(\mathsf{M}\vex k_0,\mathsf{M}\vex k_1),
\end{align}
Therefore, using $[\overline W_\ell(\vex k_0,\vex k_1)]^\dagger = \overline W_{-\ell}(\vex k_1,\vex k_0)$, we have for the Wilson loop operator starting at mirror-invariant momenta $\mathsf{M}_j\vex k_{*j} = \vex k_{*j}$,
\begin{align}\label{eq:MWM}
    {M}_j W_{j}(\vex k_{*j})M_j = [W_{j}(\vex k_{*j})]^\dagger.
\end{align}
So, it seems the Wannier Hamiltonian anti-commutes with the mirror operator and, therefore, must be characterized by a winding number with the mirror operators as the chiral operators. However, caution is necessary in defining such invariants since the relevant operators are defined in a projected subspace. Indeed, since $A^W = PA^WP$, it cannot simply anticommute with any operator in the full Hilbert space. This is because of the zero eigenvalues of $A^W$ in the projected space of $1-P$, which cause the winding of $A^W$ to vanish. Instead we find that the Wannier Hamiltonian matrix $\mathbb{A}^W_j(\vex k_{*j})$ anticommutes with the projected mirror operator matrix 
\begin{equation}
[\mathbb{M}_{j}(\vex k)]^{nm} = \braket{u_n(\mathsf{M}_j\vex k)|M_j u_m(\vex k)},    
\end{equation}
at $\vex k = \vex k_{*j}$. Therefore, the chiral operator itself is momentum-dependent.  We note that the local gauge symmetry in the occupied bands is expressed in terms of the Berry connection satisfying Eq.~\eqref{eq:A'A}. By contrast, the momentum-dependent chiral symmetry transforms as $\mathbb{M}_j(\mathbf{k}_{*j}) \mapsto \mathbb{g}^\dagger(\mathbf{k}_{*j}) \mathbb{M}_j(\mathbf{k}_{*j}) \mathbb{g}(\mathbf{k}_{*j})$. Therefore, it represents a projected symmetry rather than a local gauge.

We can try to address this challenge, for example, by extending the Wilson loop operator to the full Hilbert space, where we can use the momentum-independent mirror operators. A simple extension such as replacing $W_j(\vex k)$ with $W_j(\vex k) + 1-P(\vex k)$, which adds an identity matrix in the unoccupied bands at $\vex k$, does not work since the Wannier Hamiltonian in the unoccupied bands vanishes (or is constant for a general branch cut of the logarithm) and, again, cannot simply anticommute with $M$. 

A more promising extension is to replace $W_j(\vex k)$ with $W_j(\vex k)+e^{i\beta}CW_j(\vex k)C$ that supplements it with the Wilson loop operator of the unoccupied bands, allowing for a phase shift $e^{i\beta}$. For Eq.~\eqref{eq:MWM} to hold for this extended Wilson loop, we must have $e^{i\beta}=\pm1$. The Wannier Hamiltonian can now be defined in the full Hilbert space and is block-diagonal in the occupied-unoccupied subspaces. 

While this extended Wannier Hamiltonian properly anticommutes with the momentum-independent mirror operator in the full Hilbert space, its chiral winding number with respect to $M_j$ vanishes. This is because for $e^{i\beta}=\pm1$, the Wannier Hamiltonian commutes  with the chiral operator $C$ up to a constant $\beta$. Therefore, by Theorem~\ref{th:uniS} in the Appendix, the winding number vanishes. We have also confirmed this result numerically.

Since $M_j$ is itself block-diagonal in the occupied-unoccupied subspaces, this means that, with respect to $M_j$, the winding of the Wannier Hamiltonian in the occupied subspace is the opposite of its winding in the unoccupied subspace. This may be anticipated from the fact that the unoccupied Wilson loop is obtained by mapping the occupied Wilson loop by the chiral operator $C$, which in turn anticommutes with the mirror operators. Therefore, the mirror eigenspaces of the mirror operator in the occupied bands are mapped, by the chiral operator, to the eigenspaces of the mirror operator with the opposite eigenvalue in the unoccupied bands. This forces the sum of the windings of the occupied and unoccupied Wilson loops to vanish.

Thus, in order to find a nonzero value of the winding, we must \emph{subtract} the two windings. This is similar to the mirror grading used for diagonal mirror operators in $C_4$ symmetric models. To avoid confusion, we call this new invariant the ``mirror-filtered'' winding number.

Equivalently, the mirror-filtered winding can be viewed as the winding of the extended Wannier Hamiltonian with respect to a chiral operator with opposite signs in the occupied and unoccupied subspaces. However, since the choice of these subspaces is momentum-dependent, this means this chiral operator is itself momentum-dependent. Indeed, it is nothing but $M_j\underline{H}(\vex k)$, where $\underline{H}(\vex k)$ is the flattened Hamiltonian. 

Defining a winding number with respect to a momentum-dependent chiral operator requires care. In particular, such a definition requires a judicious choice  of gauge for chiral eigenspaces. In fact, even when the chiral operator is momentum-independent, some expressions of the winding number require a careful choice of gauge. To see this, let us rewrite the winding number as $w = \frac1{2\pi}(\gamma^+-\gamma^-)$, where $\gamma^\pm = \sum_\alpha \oint \braket{a^\pm_\alpha(k)|i\partial_k a^\pm_\alpha(k)} dk$ are the Berry phases accumulated by the chiral eigenstates $C\ket{a^\pm_\alpha} = \pm \ket{a^\pm_\alpha}$ (see Theorem~\ref{th:gpm} in the Appendix for a proof). Since Berry phases are gauge-dependent, their difference can only be gauge-independent if the gauge for $\gamma_\pm$ is chosen consistently. Indeed, we must choose $\ket{a^-_\alpha(k)} := \underline H(k) \ket{a^+_\alpha(k)}$ in this expression. We provide some cautionary examples in the Appendix for momentum-dependent chiral operators. 

The momentum-dependent chiral symmetry of the Wilson loops was identified~\cite{Zhu_2021} as a basis for defining winding numbers corresponding to the mid-gap bound states of Wannier spectra in cylindrical boundary conditions. As we have noted, the choice of gauge of the chiral eigenspaces in this construction is tricky and not obvious. Here we establish a gauge-independent procedure for choosing the chiral eigenstates.

We use parallel transport via Wilson line operators to define 
$\ket{u^{\pm}_{j}(\vex k)} := e^{i\zeta^\pm_j(\vex k, \vex k_0)} \overline{W}_{\bar j}(\vex k,\vex k_0) \ket{u^{\pm}_{j}(\vex k_0)}$, where $\ket{u^{\pm}_{j}(\vex k_0)}$ are eigenstates of the mirror operator $M_j$ in the occupied subspace at a fixed momentum $\vex k_0$, and phases $e^{i\zeta^\pm_j(\vex k, \vex k_0)}$ are determined below.
Given that $M_j \overline{W}_{\bar j}(\vex k, \vex k_0) M_j = \overline{W}_{\bar j}(\mathsf{M}_j\vex k, \mathsf{M}_j\vex k_0)$ and
$M_jP(\vex k) = P(\mathsf{M}_j\vex k)M_j$, and requiring $e^{i\zeta^\pm_j(\vex k, \vex k_0)} = e^{i\zeta^\pm_j(\mathsf{M}_j\vex k, \mathsf{M}_j \vex k_0)}$, we have 
$M_{j} \ket{u^{\pm}_{j}(\vex k)} = \pm \ket{u^{\pm}_{j}({\mathsf{M}_j \vex k})}$. 
In particular, we define the ``mirror-filtered'' basis along mirror-invariant momenta satisfying $M_{j} \ket{u^{\pm}_{j}(\vex k_{*j})} = \pm \ket{u^{\pm}_{j}(\vex k_{*j})}$. At the initial point $\vex k_0$, we may use the eigenvectors of the mirror operator $M_j \ket{m_{j}^\pm} = \pm\ket{m_{j}^\pm}$ to set
$\ket{u^{\pm}_{j}(\vex k_0)} := \frac1{\sqrt{\bra{m_{j}^\pm}P(\vex k_0) \ket{m_{j}^\pm}}}  P(\vex k_0) \ket{m_{j}^\pm}$. Alternatively, we may set $\ket{u^{\pm}_{j}(\vex k_0)}$ as the eigenstates of the Wilson loop $W_{\bar j}(\vex k_0)$, such that  $W_{\bar j}(\vex k_0)\ket{u^\pm_j(\vex k_0)} = \exp({i\lambda^\pm_{\bar j}}) \ket{u^\pm_j(\vex k_0)}$. We set $\zeta^\pm_j(\vex k, \vex k_0) = \frac1{2\pi}(k_{\bar j}-k_{0\bar j}) \lambda^\pm_{\bar j}$ to ensure the basis is periodic. The mirror-filtered basis $\ket{u^\pm_j(\vex k_{*j})}$ form a smooth, gauge-independent, and periodic basis for the occupied subspace along mirror-invariant momenta. 
In our numerics, we choose $k_{0j} = -\pi$.

Now, defining the projected mirror operators $M_{Pj}(\vex k) = P(\mathsf{M}_j\vex k) M_j P(\vex k)$, the mirror-filtered winding numbers can be defined by
\begin{equation}
\eta_{j} = w_{\ell_{\bar j}}[{A}_{j}^W(\vex k_{*j}), {M}_{Pj}(\vex k_{*j})],
\end{equation}
and calculated using the projected mirror matrix, $\mathbb{M}_j$, and the projected Wannier Hamiltonian matrix $\mathbb{A}_j^W$ in the mirror-filtered basis, where $\mathbb{M}_j$ is diagonal and $\mathbb{A}_j^W$ is block off-diagonal. Therefore, $\eta_{j}$ is the winding of the off-diagonal block of $\mathbb{A}_j^W$ in the mirror-filtered basis.

\begin{figure}
    \centering
    \includegraphics[width=\linewidth]{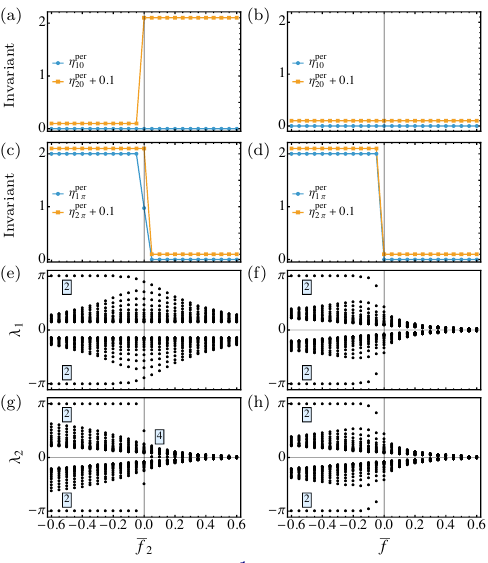}
    \caption{(a)-(d) Bulk-boundary correspondence in the separable model in Eq.~\eqref{eq:ff} ($\delta f_1 = \delta f_2=0$) showing the bulk invariants , and (e)-(h) the Wannier spectra with cylindrical boundaries with 30 unit cells in the open direction. The parameters correspond to those in Fig.~\ref{fig:BEWsep}. In (a)-(d), some values are offset slightly to distinguish the coinciding values.  In (e)-(h), the boxed numbers show the number of midgap Wannier states.}
    \label{fig:BECsep}
\end{figure}

\begin{figure}
    \centering
    \includegraphics[width=\linewidth]{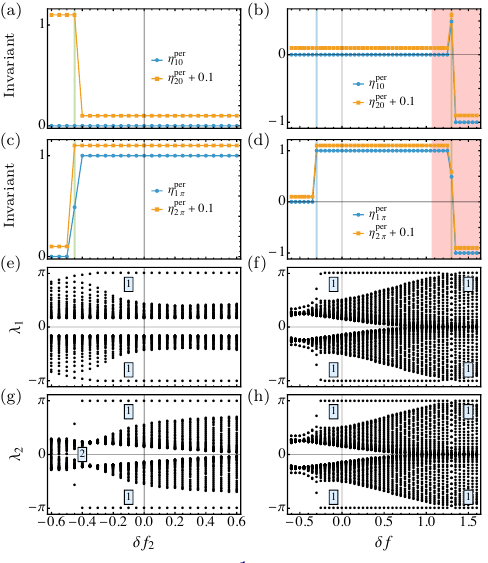}
    \caption{(a)-(d) Bulk-boundary correspondence in the non-separable model in Eq.~\eqref{eq:ff} showing the bulk invariants, and (e)-(h) the Wannier spectra with cylindrical boundaries with 30 unit cells in the open direction. The parameters correspond to those in Fig.~\ref{fig:BEWnonsep}. In (a)-(d), some values are offset slightly to distinguish the coinciding values.  In (e)-(h), the boxed numbers show the number of midgap Wannier states.}
    \label{fig:BECnonsep}
\end{figure}

\subsection{Invariants of the Wannier spectrum}

The mirror-filtered winding numbers can be calculated for loops $\ell_{\bar j}$ starting at $\vex k_{*j}$ with $k_{*j} = 0$ and $\pi$, the two invariant lines of $M_j$. Thus, we find four independent invariants $\eta_{j}(k_{*j})$ that characterize the topology of the Wannier Hamiltonians. It is natural to ask what the relationship is between these invariants and the Wannier gap structure studied in Section~\ref{sec:Wgap}.

The authors of Ref.~\cite{Zhu_2021} made an analogy between the invariants $\eta_{j}(k_{*j})$ and the invariants of the chiral Floquet Hamiltonians~\cite{Asboth_2014} by postulating the existence of a ``faithful gauge'' that yields $\frac12[\eta_{j}(0) + e^{i\varepsilon} \eta_{j}(\pi)]$ as the invariant related to the gap $\varrho_{j\varepsilon}$. While this relationship is argued to hold in the separable model studied in Ref.~\cite{Zhu_2021}, it is not clear if it holds more generally. 

As in the previous subsection, the difficulty arises due to the $k$-dependent chiral operator $\mathbb{M}_j(\vex k)$ needed to define $\eta_{j}$, since the literature on Floquet topology has only considered constant chiral operators. Here, we solve this problem by defining invariants $\eta_{j\varepsilon}$ directly related to the gap $\varrho_{j\varepsilon}$.

We start by setting an abbreviated notation for clarity. We will write $\overline{W}_j(k)$ and its projected matrix $\overline{\mathbb{W}}_j(k)$ for the Wilson line $\overline{W}_j(\vex k_{0},\vex k_1)$ where the starting and end points have $k_{0j} = 0$ and $k_{1j}=k$, respectively. We will also write $M_{j}(k)$ and its projected matrix $\mathbb{M}_j(k)$ for the projected mirror operator along $\ell_j$ at $\vex k$ with $k_{j}=k$. To avoid confusion with the original mirror operator $M_j$ we will always show the argument in the following.

Next, we define a periodic Wilson line
\begin{equation}\label{eq:Wper}
    \overline{W}_{j\varepsilon}^\text{per}(k) := e^{-i(k/2\pi)A^W_{j\varepsilon}} \overline{W}_j(k),
\end{equation}
where the dependence on $\varepsilon$ comes from the branch cut of $\ln_\varepsilon$ at $e^{i\varepsilon}$ used to define $A^W_{j\varepsilon} := -i\ln_\varepsilon W_j$. It is straightforward to show $\overline{W}_j^\text{per}(k+2\pi) = \overline{W}_j^\text{per}(k)$. A branch cut at $e^{i\varepsilon}$ can be chosen for $\ln_\varepsilon W_j$ if and only if there is a gap $\varrho_{j\varepsilon}\neq0$ in the Wannier spectrum of $W_j$. 

In close parallel to the definitions used for the chiral Floquet Hamiltonians in Ref.~\cite{Yao_2017}, we have
\begin{equation}
    M_j(0) A^W_{j\varepsilon} M_j(0) = - A^W_{j,-\varepsilon} + 2\pi.
\end{equation}
From this and using using Eq.~\eqref{eq:MWlineM}, it follows
\begin{equation}
    M_j(0) \overline{W}^\text{per}_{j\varepsilon}(k) M_j(k) = e^{-ik} \overline{W}^\text{per}_{j,-\varepsilon}(-k).
\end{equation}
This is a new result for a $k$-dependent chiral operator.

Therefore, we have for $\varepsilon = 0$ and $\pi$,
\begin{equation}\label{eq:MWMe}
    M_j(0) \overline{W}^\text{per}_{j\varepsilon}(\pi) M_j(\pi) = - e^{i\varepsilon} \overline{W}^\text{per}_{j\varepsilon}(\pi),
\end{equation}
where we have used the periodicity of $\overline W^\text{per}$ and the fact that $\ln_{\varepsilon+2\pi} = \ln_{\varepsilon} + 2\pi i$.
Now, since $M_j(0)$ and $M_j(\pi)$ have the same $\pm1$ eigenvalues, they can be diagonalized to the same diagonal matrix $D = S_j^\dagger(k_{*j}) M_j(k_{*j}) S_j(k_{*j})$ with $D^2=1$. Replacing this in Eq.~\eqref{eq:MWMe}, we find
\begin{equation}\label{eq:DSWSe}
    D S_j^\dagger(0) \overline{W}^\text{per}_{j\varepsilon}(\pi) S_j(\pi) = - e^{i\varepsilon} S_j^\dagger(0) \overline{W}^\text{per}_{j\varepsilon}(\pi) S_j(\pi) D.
\end{equation}
Since $D$ is diagonal, this means the $S_j^\dagger(0) \overline{W}^\text{per}_{j\varepsilon}(\pi) S_j(\pi)$ is block diagonal for $\varepsilon=\pi$ and block off-diagonal for $\varepsilon=0$. This is also a new result for a $k$-dependent chiral operator.

Thus, we define the invariants $\eta^\text{per}_{j\varepsilon}$ as the winding number of $D_{j\varepsilon}$, the diagonal or off-diagonal block of $S_j^\dagger(0) \overline{W}^\text{per}_{j\varepsilon}(\pi) S_j(\pi)$, respectively, for $\varepsilon=\pi$ or $\varepsilon=0$:
\begin{equation}
    \eta^\text{per}_{j\varepsilon} = \frac1{2\pi i} \oint_{\ell_{\bar j}} \partial_k \ln\det[D_{j \varepsilon}(k)] dk,
\end{equation}
where we have restored the dependence on the momentum component along $\ell_{\bar j}$. In concrete calculations, we can work directly with the projected matrices of Wilson lines and mirror operators, so that $\overline{\mathbb{W}}^\text{per}_{j\varepsilon}(k) = e^{-i(k/2\pi)\mathbb{A}^W_{j\varepsilon}}\overline{\mathbb{W}}_j(k)$, $\mathbb{M}_j(k_{*j})$ are diagonalized by a unitary $\mathbb{S}_j(k_{*j})$, and $\mathbb{D}_{j\varepsilon}$ is the corresponding block of $\mathbb{S}_j^\dagger(0) \overline{\mathbb{W}}^\text{per}_{j\varepsilon}(\pi) \mathbb{S}_j(\pi)$. In order to calculate the diagonalizing unitaries $\mathbb{S}_j$, we also need to fix the gauge of the projected basis. We do so in the same way as before, i.e. by using the mirror-filtered basis.  In effect, $\mathbb{S}_j^\dagger(0) \overline{\mathbb{W}}^\text{per}_{j\varepsilon}(\pi) \mathbb{S}_j(\pi)$ is exactly the matrix representation of $\overline{W}^\text{per}_{j\varepsilon}(\pi)$ in the mirror-filtered basis.

Now, we are in a position to directly calculate and compare the mirror-filtered invariants of the Wannier Hamiltonian, $\eta_j(k_{*j})$, and the invariants defined for the Wannier spectrum, $\eta_{j\varepsilon}^\text{per}$. We show the values of the latter in Figs.~\ref{fig:BECsep}-\ref{fig:BEC-BC} and report that in all the cases we have studied, we have found that
\begin{equation}
\eta_{j\varepsilon}^\text{per} = \frac{\eta_{j}(0) + e^{i\varepsilon}\eta_j(\pi)}{2},
\end{equation}
holds true.

\begin{figure}
    \centering
    \includegraphics[width=\linewidth]{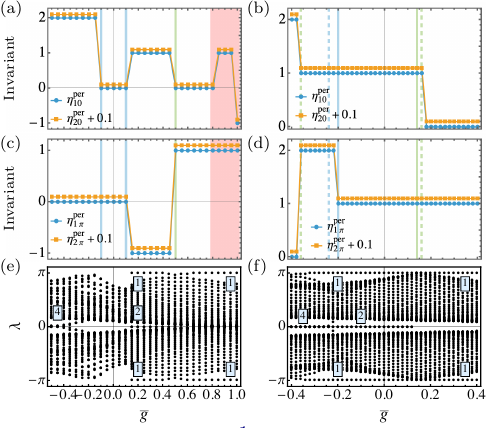}
    \caption{(a)-(d) Bulk-boundary correspondence in the for the $C_4$ symmetric model~\eqref{eq:bc} with variable $\bar g$ showing the bulk invariants, and (e), (f) the Wannier spectra with cylindrical boundaries with 30 unit cells in the open direction. The parameters correspond to those in Fig.~\ref{fig:BEW-BC}. In (a)-(d), some values are offset slightly to distinguish the coinciding values.  In (e) and (f), the boxed numbers show the number of midgap Wannier states.}
    \label{fig:BEC-BC}
\end{figure}

\begin{figure}
    \centering
    \includegraphics[width=\linewidth]{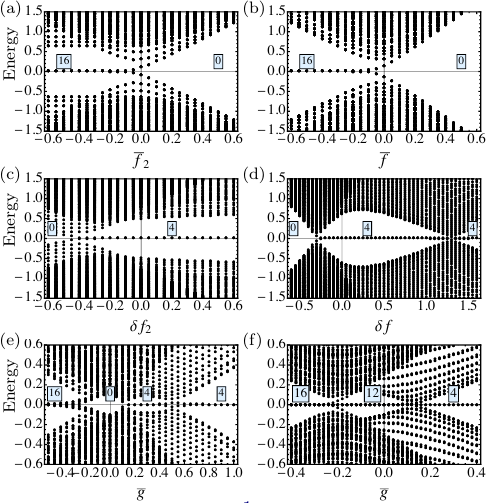}
    \caption{ Energy spectra with open boundaries and $30$ unit cells in each direction for (a),(b) the separable model in Eq.~\eqref{eq:ff} with parameters corresponding, respectively, to those in Fig.~\ref{fig:BEWsep} left and right panels; (c),(d) nonseparable model in Eq.~\eqref{eq:ff} with parameters corresponding, respectively, to those in Fig.~\ref{fig:BEWnonsep} left and right panels; and (e),(f) the $C_4$ symmetric model in Eq.~\eqref{eq:bc} with parameters corresponding, respectively, to those in Fig.~\ref{fig:BEW-BC} left and right panels. The boxed numbers show the
number of zero-energy corner-bound states.}
    \label{fig:OBC}
\end{figure}

We observe that the Wannier gap closings and nonanalytic changes in $\varrho_{j\varepsilon}$ indeed result in changes to $\eta^\text{per}_{j\varepsilon}$, thus fully characterizing the topology of Wannier spectra, and furnishing a bulk-boundary correspondence for Wilson loops. 
 We investigate the bulk-boundary correspondence by calculating the Wannier spectra, obtained as phases $\lambda_j$ of the eigenvalues $e^{i\lambda_j}$ of the Wilson loops $W_j$, under cylindrical boundary conditions. These are shown in Fig.~\ref{fig:BECsep}(e)-(h), Fig.~\ref{fig:BECnonsep}(e)-(h) and Fig.~\ref{fig:BEC-BC}(e),(f). In nearly all cases the mirror-filtered invariants $\eta_{j\varepsilon}^\text{per}$ match the number of midgap Wannier states in the corresponding gaps at $\varepsilon = 0,\pi$. The one exception is found in Fig.~\ref{fig:BEC-BC} right panels, where for $-0.4\lesssim \bar g \lesssim -0.2$, we find $\eta^\text{per}_{j0}=1$ and $\eta^\text{per}_{j\pi}=2$ but observe two midgap states at $\varepsilon=0$ and one midgap state at $\varepsilon=\pi$ per edge of the cylinder. We do not know the source of this discrepancy at this time. Additionally, for the $C_4$ symmetric models (the right panels in Figs.~\ref{fig:BECsep} and~\ref{fig:BECnonsep} and in Fig.~\ref{fig:BEC-BC}), we find $\eta_{1\varepsilon}^\text{per}=\eta_{2\varepsilon}^\text{per}$ and $\lambda_1 = \lambda_2 \equiv \lambda$, as expected. Interestingly, the Wilson loop invariants change not only at Wannier gap closings, but also at bulk energy gap closings that correspond to kinks in the Wannier gaps shown in Figs.~\ref{fig:BEWsep}-\ref{fig:BEW-BC}. These changes correspond to the boundary Wannier states under cylindrical boundary conditions, which can appear and disappear even without Wannier gaps closings.

However, the relation to the bulk-boundary correspondence of the original Hamiltonian is less clear.  In Fig.~\ref{fig:OBC}, we show the energy spectra under open boundary conditions for the separable and nonseparable models. We have constructed various combinations of invariants $\eta^\text{per}_{j\varepsilon}$, motivated by different intuitions, that match the number of corner-bound states of the original Hamiltonian in open boundary conditions in some but not all cases.  For example, it is easy to check that the product $\eta^\text{per}_{1\pi}\eta^\text{per}_{2\pi}$ does so by comparing Figs.~\ref{fig:BECsep} and \ref{fig:BECnonsep} to Fig.~\ref{fig:OBC}(a)-(d).
This can be motivated by the expectation that changes in the invariants characterizing the gap $\varrho_{j\pi}$ are related to the corresponding spectral edge gap closing, and then by the intuition that the edge theory maintains the tensor-product form of the bulk Hamiltonian.  However, this construction fails to capture the bulk-boundary correspondence in Fig.~\ref{fig:BEC-BC} compared to Fig.~\ref{fig:OBC}(e) and (f). We have also inspected more complicated constructions, such as $\min(\eta_{\max\times},\eta_{\max+})$ where $\eta_{\max\circ} = \max(\eta_{10},\eta_{1\pi}) \circ \max(\eta_{20},\eta_{2\pi})$, which captures the correct number of corner-bound states in all cases we studied except the right panels of Fig.~\ref{fig:BEC-BC}. The motivation behind this construction is more complicated, but since it does not work in all the cases, we only present it as an interesting attempt at such a construction.

\section{Summary and outlook}\label{sec:sum}

In this work, we explored the higher-order bulk–boundary correspondence of a family of chiral-symmetric Bloch Hamiltonians incorporating mirror and rotational symmetries. These models extend the $\pi$-flux square model, the prototypical topological quadrupole insulator, to both separable and nonseparable models with extended and mixed hopping. Previous work on separable models established a $\mathbb{Z}$-valued invariant governing the number of corner states. However, this invariant ceases to exist in nonseparable models, motivating the development of new momentum-space diagnostics.

To address this challenge, we introduced gauge-independent mirror-filtered winding numbers for Wannier Hamiltonians, defined via Wilson loops in a mirror-projected basis of the occupied states. We further constructed a complementary family of invariants associated with Wannier gaps at $\varepsilon = 0$ and $\pi$, using periodicized Wilson lines adapted from Floquet theory. These invariants provide a detailed characterization of Wilson-loop and Wannier-band topology. While they furnish a bulk-boundary correspondence for Wilson loops, their relation to the physical higher-order boundary spectrum is not universal, underscoring a subtle distinction between Wannier-sector topology and the boundary physics of the underlying Hamiltonian.

Several questions and new directions emerge from our results.  Our construction of mirror-filtered and periodicized Wilson-loop invariants can be extended to other symmetry classes, including those with antiunitary or nonsymmorphic crystalline symmetries. An open challenge is the formulation of momentum-space invariants for chiral higher-order phases in general nonseparable models. While real-space invariants such as Bott indices~\cite{Ono_2019,Trifunovic_2019b,Wheeler_2019,Agarwala_2020,Lin_2020,Kang_2019,Lin_2021,Benalcazar_2022} are useful, especially in the presence of disorder, it is still desirable to construct momentum-space invariants using the Bloch Hamiltonian. This not only provides an additional diagnostic for higher-order topological phases, it can also open new avenues for exploring extensions of such phases analytically.

\begin{acknowledgments}
We acknowledge fruitful discussions with T. L. Hughes and correspondence with W. Benalcazar and A. Cerjan in various stages of this work. This work is supported in part by Indiana University Office for Research Development (Bridge Program), IU Institute for Advanced Study, and by the National Science Foundation through Grant No. DMR-2533543.
\end{acknowledgments}

\appendix*

\section{Chiral winding numbers}

\subsection{Reformulations of the winding number}
Consider the Hamiltonian $H(k)$ with chiral operators $C$, $\{H,C\}=0$. In the chiral basis, where $C=\sigma_z\otimes 1$ we have $H(k) =\sigma_x\otimes d_1(k) + \sigma_y\otimes d_2(k)$ with Hermitian operators $d_1$ and $d_2$. The winding number of the Hamiltonian $H$ relative to the chiral operator $C$ is defined as
\begin{equation}
    w[H,C] = \frac1{2\pi i}\oint \partial_k\left[ \ln\det d(k)\right] dk,
\end{equation}
where $d(k) = d_1(k) - i d_2(k)$.

\begin{WNth}\label{th:qdq}
The winding number can be expressed as
\begin{equation}
w[H,C] = \frac1{2\pi i}\oint \mathrm{tr}\left[q^\dagger(k) \partial_k q(k)\right] dk,
\end{equation}
where $q$ is the unitary obtained by flattening $d$, i.e. $q(k) = u(k)v^\dagger(k)$ given the singular-value decomposition $d(k) = u(k)D(k)v^\dagger(k)$.
\end{WNth}

\begin{proof}
We have $\det d(k) = \det q(k) \det D(k)$, so that $\partial_k\ln\det d(k) = \partial_k\ln\det q(k) + \partial_k \ln\det D(k)$. Since $\ln\det D(k)$ is real and periodic, it has no winding. Then, using Laplace’s identity for the derivative of the determinant, $\partial_k\det h = \text{tr}\left(h^{-1}\partial_k h\right)\det h$, and the unitarity of $q(k)$, we find the final result.
\end{proof}

\begin{WNth}\label{th:Cudu}
Another expression for the winding number is
\begin{equation}
    w[H,C] = \frac{1}{\pi } \sum_\alpha \oint\braket{Cu_\alpha(k)|i\partial_ku_\alpha(k)}dk,
\end{equation}
where $\ket{u_\alpha(k)}$ are eigenstates of $H(k)$ with positive or negative eigenvalues.
\end{WNth}

\begin{proof}
In the chiral basis, $C=\sigma_z\otimes I$, we may write
\begin{equation}
    \ket{u_\alpha} = \frac1{\sqrt 2} \begin{bmatrix} \ket{a_\alpha} \\ \ket{b_\alpha} \end{bmatrix}, \quad C\ket{u_\alpha} = \frac1{\sqrt 2} \begin{bmatrix} \ket{a_\alpha} \\ -\ket{b_\alpha} \end{bmatrix}.
\end{equation}
The orthonormality of the basis set $\{ \ket{u_\alpha}, C\ket{u_\alpha} \}$ yields
\begin{align}
    \sum_\alpha\ket{a_\alpha}\bra{a_\alpha} &= \sum_\alpha\ket{b_\alpha}\bra{b_\alpha} = 1, \\
    \braket{a_\alpha|a_\beta} &=\braket{b_\alpha|b_\beta} = \delta_{\alpha\beta}.
\end{align}
That is, the states $\{\ket{a_\alpha}\}$ and $\{\ket{b_\alpha}\}$ furnish orthonormal bases for the projected chiral eigenspaces with $\pm1$ eigenvalues, respectively. We write the flattened Hamiltonian as
\begin{align}
    \underline H 
    &= \sum_\alpha \big(\ket{u_\alpha}\bra{u_\alpha}-C\ket{u_\alpha}\bra{u_\alpha}C \big) \\
    &= \begin{bmatrix} 0 & \sum_\alpha \ket{a_\alpha}\bra{b_\alpha} \\ \sum_\alpha \ket{b_\alpha}\bra{a_\alpha} & 0 \end{bmatrix}.    
\end{align}
So, the off-diagonal unitary block $q=\sum_\alpha \ket{a_\alpha}\bra{b_\alpha}$ maps $\ket{b_\alpha} \mapsto \ket{a_\alpha}$.
Therefore,
\begin{align*}
w[H,C] &= \frac1{2\pi i} \oint \text{tr}[q^\dagger(k) \partial_k q(k)]dk \\
&= \frac1{2\pi} \sum_\alpha \oint [\braket{a_\alpha|i\partial_k a_\alpha} - \braket{b_\alpha|i\partial_k b_\alpha}] dk \\
&= \frac1{\pi} \sum_\alpha \oint \bra{u_\alpha}i\partial_k \ket{Cu_\alpha} dk \\
&= \frac1{\pi} \sum_\alpha \oint \braket{C u_\alpha|i\partial_k u_\alpha} dk.
\end{align*}
\end{proof}

The states $\ket{u_\alpha}$ and $C\ket{u_\alpha}$ have opposite eigenvalues of $H$. Note that once a gauge has been fixed for the eigenspace spanned by $\{\ket{u_\alpha}\}$, the gauge for the eigenspace with the opposite eigenvalue is fixed by $C$. Also, the winding number is independent of the choice of a smooth gauge for eigenspace spanned by $\{\ket{u_\alpha}\}$. Fixing the gauge also fixes the gauge for the chiral eigenspaces spanned by $\{\ket{a_\alpha}\}$ and $\{\ket{b_\alpha}\}$. This means gauge transformations $\ket{a_\alpha(k)} \mapsto g_a\ket{a_\alpha(k)}$ and $\ket{b_\alpha(k)} \mapsto g_b \ket{b_\alpha(k)}$ can differ only by a $k$-independent unitary $g$: $g_a(k) = g\,g_b(k)$. Therefore, $\oint\text{tr}[g_a^\dagger\partial_k g_a]dk = \oint\text{tr}[g_b^\dagger\partial_k g_b]dk$ and their contributions to the winding number cancel out.

\begin{WNth}\label{th:gpm}
The winding number can be expressed as
\begin{align}\label{eq:wgpm}
w[H,C] = \frac{\gamma^+-\gamma^-}{2\pi},
\end{align}
where  $\gamma^\pm = \sum_\alpha \oint \braket{a^\pm_\alpha(k)|i\partial_k a^\pm_\alpha(k)} dk$, $\ket{a^-_\alpha(k)} := \underline H(k) \ket{a^+_\alpha(k)}$, and $\{\ket{a^+_\alpha(k)}\}$ is a continuously defined frame for the positive chiral subspace and $\underline H(k)$ is the flattened Hamiltonian.
\end{WNth}

\begin{proof}
Using the notation introduced in the proof of Theorem~\ref{th:Cudu}, in the chiral basis
\begin{equation}
\ket{a^+_\alpha} = \begin{bmatrix} \ket{a_\alpha} \\ 0\end{bmatrix}, \quad \ket{a^-_\alpha} := \underline H \ket{a^+}= \begin{bmatrix} 0 \\ \ket{b_\alpha} \end{bmatrix}.    
\end{equation}
As a result, we may also choose the frame as $\ket{a^+_\alpha} = \frac1{\sqrt 2}(\ket{u_\alpha} + C\ket{u_\alpha})$ independent of basis. The result then follows by expanding the right-hand side of Eq.~\eqref{eq:wgpm} and using Theorem~\ref{th:Cudu}. Note that the result is gauge invariant since the gauge of $\ket{a^-_\alpha}$ is fixed by that of $\ket{a_\alpha^+}$ and cancels after subtraction.
\end{proof}

\begin{WNcr}\label{cr:PHdH}
    The winding number can be expressed as
    \begin{equation}\label{eq:PHdH}
        w[H,C] = \frac1{2\pi i}\oint \mathrm{tr}\left[P_+ \underline H(k) \partial_k \underline H(k) \right]dk,
    \end{equation}
    where $P_+$ is the projector over the positive chiral subspace and the trace is taken over the entire Hilbert space.
\end{WNcr}

\begin{proof}
We have
$$\braket{a^-_\alpha|\partial_k a^-_\alpha} = \braket{a^+_\alpha|\partial_k a^+_\alpha} +\bra{a^+_\alpha}\underline H\partial_k \underline H \ket{a^+_\alpha},$$ 
where we have used $\ket{a^-} = \underline H \ket{a^+}$.
The result follows after writing $\bra{a^+_\alpha}\underline H\partial_k \underline H \ket{a^+_\alpha} = \mathrm{tr}\left[  \ket{a^+}\bra{a^+} \underline H \partial_k \underline H\right] = \mathrm{tr}\left[P_+ \underline H \partial_k \underline H \right]$.
\end{proof}

A word of caution about the calculation of the winding number in Theorem~\ref{th:gpm} and its Corollary~\ref{cr:PHdH}: it should not be calculated simply as the winding of the off-diagonal element of $\underline H$ in the $\{\ket{a^\pm}\}$ basis. Instead, it should be calculated as the difference of Berry phases in Eq.~\eqref{eq:wgpm}, or the trace formula in Eq.~\eqref{eq:PHdH}. For example, taking $\underline H(k) = \sigma_x \cos k + \sigma_y \sin k$ with $C = \sigma_z$ and $\ket{a^+} = (1, 0)$ we have $\ket{a^-} = \underline H \ket{a^+} = (0, e^{ik} )$. In this basis $\underline H = \sigma_x$, so the off-diagonal element has no winding. In fact, this is quite general: $\underline H$ is off-diagonal with the elements $\bra{a_\alpha^+}\underline H \ket{a_\beta^-} = \bra{a_\alpha^+} \underline H^2 \ket{a_\beta^+} = \delta_{\alpha\beta}$.

The expressions in Theorem~\ref{th:gpm} and its Corollary~\ref{cr:PHdH} allow us to define winding numbers with respect to a $k$-dependent chiral operator $C(k)$, since we can still define a continuous frame $\ket{a_\alpha^+(k)}$ in the positive eigenspace of $C(k)$ and its corresponding frame $\ket{a^-_\alpha(k)} := \underline H(k)\ket{a^+_\alpha(k)}$. The winding number $w[H,C]$ is thus precisely the winding number of the unitary map $q(k): \ket{a^-} \mapsto \underline H(k)\ket{a^-}$ onto the positive chiral subspace, independent of the choice of the original frame. The only requirement on the frame is that the frame must be smooth over $k$.

However, while the winding number defined in this way is independent of the choice of the frame, it is not an integer. This is so because the chiral eigenspaces themselves may experience a winding, which manifests as non-quantized Berry phases $\gamma^\pm$. We show these properties explicitly in two examples.

First, consider $\underline H = \sigma_z$. This is, of course, constant and so it should not have any winding number. Any unitary Hermitian $c_x \sigma_x + c_y \sigma_y$ with $c_x+ic_y$ on the unit circle anticommutes with $\underline H$. What if $c_x$ and $c_y$ themselves depend on $k$? For example, we may choose $C(k) =\cos k\, \sigma_x + \sin k\, \sigma_y$. There appears to be a gauge freedom in choosing the phases of the chiral eigenstates $\ket{a^+(k)}$ and $\ket{a^-(k)}$. If we choose these phases arbitrarily, the off-diagonal element of $\underline H$ will also have an arbitrary phase in the chiral basis and the winding number appears to be gauge-dependent. However, only one of the phases is arbitrary, say that of $\ket{a^+(k)}$ and the other phase is fixed by $\ket{a^-(k)} = \underline H\ket{a^+(k)}$. Then, since $\underline H$ itself is constant the Berry phases $\gamma^\pm$ are guaranteed to be equal and the winding number $w[H,C] = 0 $ for all choices of $C$.

Second, take $\underline H(k) = \cos\theta(\cos k\, \sigma_x + \sin k\, \sigma_y) - \sin\theta\, \sigma_z$ and $C(k) = \cos\theta\,\sigma_z+\sin\theta(\cos k\, \sigma_x + \sin k\, \sigma_y)$. Then
$$
\ket{a^+(k)} = \begin{bmatrix} \cos(\theta/2) \\ \sin(\theta/2)e^{ik} \end{bmatrix} \Rightarrow \ket{a^-(k)} = \begin{bmatrix} -\sin(\theta/2) \\ \cos(\theta/2)e^{ik} \end{bmatrix}.
$$
Thus,
$$
\gamma^+ = \oint\braket{a^+|i\partial_k a^+}dk = -2\pi\sin^2(\theta/2),
$$
and
$$
\gamma^- = \oint\braket{a^-|i\partial_k a^-}dk = -2\pi\cos^2(\theta/2).
$$
Therefore, $(\gamma^+-\gamma^-)/2\pi = \cos\theta.$

While the winding number defined by Theorem~\ref{th:gpm} and its Corollary~\ref{cr:PHdH} with respect to a $k$-dependent chiral operator is independent of the choice of frame of a given chiral eigenspace, it does depend on the choice of the chiral basis along $k$, even if we insist the basis is smooth.

\begin{WNcr}[Change of basis]
Define a change of basis with a unitary $U(k)$ as
\begin{align*}
H'(k) &= U^\dagger(k) H(k) U(k) , \\
C'(k) &= U^\dagger(k) C(k) U(k) .
\end{align*}
Then, the expression in \eqref{eq:PHdH} satisfies
$$
w[H',C'] = w[H,C] - w[U^\dagger,C], 
$$
where $w[U^\dagger,C] = \frac1{2\pi i}\oint \mathrm{tr}[C U\partial_kU^\dagger] dk = -w[U,C']$ and $C := 2P_C - 1= P_C - \underline{H} P_C \underline H$.    
\end{WNcr}

\begin{widetext}
\begin{proof}
We have    
\begin{align*}
\oint \mathrm{tr} \left[ P_{C'} \underline{H}'\partial_{k}\underline{H}' \right] dk
&= \oint \mathrm{tr} \left[ P_{C'} \left( U^{\dagger} \underline{H} U \right)\partial_{k} \left(U^{\dagger} \underline{H} U \right) \right] dk \\
&= \oint \mathrm{tr} \left[ P_{C'} U^{\dagger}\underline{H}U \left\{  ( \partial_{k} U^{\dagger} )\underline{H}U
+ U^{\dagger} ( \partial_{k}\underline{H} ) U
+ U^{\dagger}\underline{H} ( \partial_{k} U ) \right\}
\right] dk
\\
&= \oint \mathrm{tr} \left[
P_{C'}  U^{\dagger}\underline{H} U  ( \partial_{k} U^{\dagger} ) \underline{H}U
+ P_{C'} U^{\dagger}\underline{H}(\partial_{k}\underline{H})U
+ P_{C'} U^{\dagger}\partial_{k}U
\right] dk
\\
&= \oint \mathrm{tr} \left[
(1-P_{C}) U  \partial_{k} U^{\dagger}
+ P_{C} \underline{H} \partial_{k}\underline{H}
+ P_{C'} U^{\dagger}\partial_{k}U
\right] dk
\\
&= \oint \mathrm{tr} \left[
(1-P_{C}) U  \partial_{k} U^{\dagger}
+ P_{C} \underline{H} \partial_{k}\underline{H}
- P_{C} U \partial_{k} U^{\dagger}
\right] dk
\\
&= \oint \mathrm{tr} \left[ P_{C} \underline{H} \partial_{k}\underline{H}  \right]- \oint \mathrm{tr} \left[
C U  \partial_{k} U^{\dagger} \right],
\end{align*}
where we have used $U P_{C'} U^\dagger = P_C$, $\underline{H} P_C \underline{H} = 1-P_C$, and $C = 2P_C-1$. The result follows after dividing by $2\pi i$.
To show $w[U^\dagger,C] = -w[U,C']$, we write $\mathrm{tr}[CU\partial_k U^\dagger] = -\mathrm{tr}[C (\partial_k U) U^\dagger] = -\mathrm{tr}[U^\dagger C U U^\dagger \partial_k U] = -\mathrm{tr}[C' U^\dagger \partial_k U]$.
\end{proof}
\end{widetext}

\subsection{Symmetry properties of the winding number}

We now establish a few results for the winding number under symmetry operations and when it vanishes.

\begin{WNth}\label{th:uniS}
The winding number vanishes for a Hamiltonian with a unitary symmetry, $S H(\vex k) S^\dagger = H(\vex k)$, that anticommutes with the chiral operator.
\end{WNth}

\begin{proof}
We have,
\begin{align*}
    \oint \mathrm{tr}[P_C \underline H \partial_k \underline H] dk
    &= \oint \mathrm{tr}[P_C S \underline H \partial_k \underline H S^\dagger] dk \\
    &= \oint \mathrm{tr}[S^\dagger P_C S \underline H \partial_k \underline H] dk \\
    &= - \oint \mathrm{tr}[P_C \underline H \partial_k \underline H] dk,
\end{align*}
where we have used $S^\dagger P_C S = 1-P_C$, and the fact that $\oint \tr[\underline H \partial_k \underline H]dk = \frac12\oint \mathrm{tr}[\partial_k \underline H^2] = 0$. Thus, $w[H, C] = -w[H, C] =0$.
\end{proof}

\begin{WNcr}\label{cr:C1C2}
The winding number vanishes for a Hamiltonian with two anticommuting chiral operators, with respect to both chiral operators.
\end{WNcr}

\begin{proof}
    Take $S = C_1 C_2$ and use Theorem~\ref{th:uniS}.
\end{proof}

\begin{WNth}\label{th:antiS}
The winding number vanishes for a Hamiltonian with an antiunitary symmetry, $\Sigma H(\vex k) \Sigma^\dagger = H(\vex k)$, that commutes with chiral operator.
\end{WNth}

\begin{proof}
We have
\begin{align*}
    \oint \mathrm{tr}[P_C \underline H \partial_k \underline H] dk
    &= \oint \mathrm{tr}[P_C \Sigma \underline H \partial_k \underline H \Sigma^\dagger] dk \\
    &= \oint \mathrm{tr}[\underline H \partial_k \underline H \Sigma^\dagger P_C \Sigma ]^* dk \\
    &= \oint \mathrm{tr}[P_C \underline H \partial_k \underline H]^*dk,
\end{align*}
where have used the fact that for two antilinear operators $\Sigma_1$ and $\Sigma_2$, $\mathrm{tr}[\Sigma_1\Sigma_2] = \mathrm{tr}[\Sigma_2\Sigma_1]^*$. Since the winding number is real, $w[H, C] = -w[H,C]^* = 0$.
\end{proof}

\vspace{-2.5mm}
For our Hamiltonian~\eqref{eq:H0}, there is an antiunitary symmetry $\Sigma = I\Theta$, the product of time-reversal operator and inversion $I=iM_2M_1$ that commutes with the chiral operator. Therefore, all winding numbers $w_\ell[H,C]$ vanish.

\begin{WNth}
For a symmetry $M H(q)M = H(Mq)$, $M^2 = 1$ that commutes or anticommutes with the chiral symmetry, respectively,
\begin{equation}
  w_\ell[H,C] =  \pm w_{M\ell}[H,C].
\end{equation}
\end{WNth}

\begin{proof}
Noting $q\in \ell \mapsto Mq\in M\ell$ we have
\begin{align*}
w_\ell[H,C] 
	&= \frac1{2\pi i} \oint_{M\ell} \tr[P_C M \underline H M \partial_k M\underline H M] dk \\
	&= \frac1{2\pi i} \oint_{M\ell} \tr[M P_C M \underline H \partial_k \underline H] dk \\
	&= \pm \frac1{2\pi i} \oint_{M\ell} \tr[P_C \underline H \partial_k \underline H] dk \\
	&= \pm w_{M\ell}[H,C],
\end{align*}
where, in the penultimate line we have used $MP_CM = P_C$ for commuting symmetry and $MP_CM = 1-P_C$ for anticommuting symmetry, respectively, and $\oint \tr[\underline H \partial_k \underline H]dk = 0$.
\end{proof}

\begin{WNcr}\label{cr:MC}
The winding number vanishes for a Hamiltonian with a symmetry that commutes or anticommutes with the chiral operator, respectively, on a loop that is inverted or invariant, respectively, under symmetry operation.
\end{WNcr}

\begin{proof}
We have, respectively $M\ell = \mp \ell$. Thus, $w_\ell[H,C] = \pm w_{\mp\ell}[H,C] = - w_\ell[H,C] = 0$.
\end{proof}

For our Hamiltonian~\eqref{eq:H0}, we have $m_{12}c_1 = c_1 m_{12}$ and $m_{12}\ell_2 = -\ell_2$, thus $w_{\ell_2}[h_1,c_1] = 0$. Similarly, $m_{21}c_2 = c_2 m_{21}$ and $m_{21}\ell_1 = -\ell_1$, thus $w_{\ell_1}[h_2,c_2] = 0$.

\bibliography{refs}

\end{document}